\documentclass[12pt]{article}

\usepackage{amssymb}
\usepackage{amstext}
\usepackage{amsfonts}
\usepackage{amsmath}
\usepackage{graphicx}
\usepackage{fullpage}

\newtheorem{theorem}{Theorem}[section]    % Specify Theorem
\newtheorem{fact}[theorem]{Fact}    % Specify Fact
\newtheorem{definition}{Definition}[section] % Specify Definition
\newtheorem{claim}[theorem]{Claim}    % Specify Claim
\newtheorem{lemma}[theorem]{Lemma}    % Specify Lemma
\newcommand{\qed}{\hfill{$\rule{6pt}{6pt}$}} %Box at end of proof
\newenvironment{proof}{\noindent{\bf Proof:}}{\qed\\}
\newenvironment{proofof}[1]{\noindent{\bf Proof of #1:}}{\qed\\}

\numberwithin{equation}{section} 

\newcommand{\floor}[1]{{\lfloor #1 \rfloor}}

\newcommand{\defeq}{\stackrel{\mathrm{def}}{=}}

\newcommand{\expct}{{\mathbb E}}

\newcommand{\suppress}[1]{}
\newcommand{\comment}[1]{}

\newcommand{\eps}{\varepsilon}

\newcommand{\sR}{{{\mathsf R}}}
\newcommand{\sD}{{{\mathsf D}}}
\newcommand{\mcX}{{{\mathcal X}}}
\newcommand{\mcY}{{{\mathcal Y}}}
\newcommand{\mcZ}{{{\mathcal Z}}}
\newcommand{\mcP}{{{\mathcal P}}}
\newcommand{\ment}{{{\mathsf{ment}}}}
\newcommand{\crent}{{{\mathsf{crent}}}}
\newcommand{\cment}{{{\mathsf{rcment}}}}
\newcommand{\err}{{{\mathsf{err}}}}
\newcommand{\pub}{{{\mathsf{pub}}}}

\title{ {\bf Strong direct product conjecture holds for all relations in public coin randomized one-way communication complexity  }}

\author{Rahul Jain\thanks{Centre for Quantum Technologies and Department of Computer Science, National University of Singapore. {\tt rahul@comp.nus.edu.sg}}
}

\begin{document}

\maketitle

\abstract{
Let $f \subseteq \mcX \times \mcY \times \mcZ$ be a relation. Let the public coin one-way communication complexity of $f$, with worst case error $1/3$, be denoted $\sR^{1,\pub}_{1/3}(f)$. We show that if for computing $f^k$ ($k$ independent copies of $f$), $o(k \cdot \sR^{1,\pub}_{1/3}(f))$ communication is provided, then the success is exponentially small in $k$. This settles the strong direct product conjecture for all relations in public coin one-way communication complexity.

We show a new tight characterization of public coin one-way communication complexity which strengthens on the tight characterization shown in J., Klauck, Nayak~\cite{JainKN08}. We use the new characterization to show our direct product result and this may also be of independent interest.
}

\section{Introduction}
\label{sec-introduction}

Let $f \subseteq \mcX \times \mcY \times \mcZ$ be a relation and $\eps >0$. Let Alice with input $x \in \mcX$, and Bob with input $y \in \mcY$, wish to compute a $z \in \mcZ$ such that $(x,y,z) \in f$. We consider the model of public coin one-way communication complexity in which Alice sends a single message to Bob, and Alice and Bob may use pubic coins. Let $\sR^{1,\pub}_\eps(f)$ denote the communication of the best protocol $\mcP$ which achieves this with error at most $\eps$ (over the public coins) for any input $(x,y)$. Now suppose that Alice and Bob wish to compute $f$ simultaneously on $k$ inputs $(x_1,y_1), \ldots, (x_k, y_k)$ for some $k \geq 1$. They can achieve this by running $k$ independent copies of $\mcP$ in parallel . However in this case the overall success could be as low as $(1 - \eps)^k$. Strong direct product conjecture for $f$ states that this is roughly the best that Alice and Bob can do. We show that this is indeed true for all relations.
\begin{theorem}
\label{thm:main}
Let $f \subseteq \mcX \times \mcY \times \mcZ$ be a relation. Let $k \geq 1$ be a natural number. Then,
$$ \sR^{1,\pub}_{1 - 2^{-\Omega(k)}}(f^k)  \geq   \Omega( k \cdot \sR^{1,\pub}_{1/3}(f)  )  \enspace . $$
\end{theorem}
We show this result by showing a new tight characterization of public coin one-way communication complexity for all relations. We introduce a new measure of complexity which we call the robust conditional relative min-entropy bound. We show that this bound is equivalent, up to constants, to $\sR^{1,\pub}_{1/3}(f)$ and use this to show the direct product result.  This bound forms lower bound on the one-way subdistribution bound of J., Klauck, Nayak~\cite{JainKN08} where they show that their bound is equivalent, up to constants, to $\sR^{1,\pub}_{1/3}(f)$. They also showed that the one-way subdistribution bound satisfies the direct product property under product distributions. 

There has been substantial prior work on the strong direct product question and the weaker direct sum and weak direct product questions in various models of communication complexity, e.g.~\cite{ImpagliazzoRW94, ParnafezRW97,  ChakrabartiSWY01, Shaltiel03, JainRS03a, KlauckSdW04, Klauck04, JainRS05a,  BeamePSW07, Gavinsky08, JainKN08, JainK09,  HarshaJMR09, BarakBCR10, BravermanR10, Klauck10}.

In the next section we provide some information theory and communication complexity preliminaries that we need. We refer the reader to the texts~\cite{CoverT91, KushilevitzN97} for good introductions to these topics respectively. In section~\ref{sec:newbounds} we introduce our new bound. In section~\ref{sec:newchar} we show that it  tightly characterizes public coin one-way communication complexity.  Finally in section~\ref{sec:dpt} we show our direct product result. 

\section{Preliminaries}
\subsection*{Information theory}
\label{sec:inf}
Let $\mcX, \mcY$ be sets and $k$ be a natural number. Let $\mcX^k$ represent $\mcX \times \cdots \times \mcX$, $k$ times. Let $\mu$ be a distribution over $\mcX$ which we denote by $\mu \in \mcX$.  We use $\mu(x)$ to represent the probability of $x$ under $\mu$. The entropy of $\mu$ is defined as $S(\mu) = - \sum_{x \in \mcX} \mu(x) \log \mu(x)$. Let $X$ be a random variable distributed according to $\mu$ which we denote by $X \sim \mu$. We use the same symbol to represent a random variable and its distribution whenever it is clear from the context. For distributions $\mu, \mu_1 \in \mcX$, $\mu \otimes \mu_1$ represents the product distribution $(\mu \otimes \mu_1)(x) = \mu(x) \otimes \mu_1(x)$ and $\mu^k$ represents $\mu \otimes \cdots \otimes \mu$, $k$ times. The $\ell_1$ distance between distributions $\mu, \mu_1$ is defined as $||\mu - \mu_1||_1 = \frac{1}{2} \sum_{x \in \mcX} |\mu(x) - \mu_1(x)|$. Let $\lambda, \mu \in \mcX \times \mcY$.  We use $\mu(x | y)$ to represent $\mu(x,y)/\mu(y)$. When we say $XY \sim \mu$ we assume that $X \in \mcX$  and $Y \in \mcY$. We use $\mu_x$ and $Y_x$ to represent $Y |~ X =x$.  The conditional entropy of $Y$ given $X$, is defined as $S(Y|X) = \expct_{x \leftarrow X} S(Y_x)$. The relative entropy between $\lambda$ and $\mu$ is defined as $S(\lambda || \mu) = \sum_{x \in \mcX} \lambda(x) \log \frac{\lambda(x)}{\mu(x)}$. We use the following properties of relative entropy at many places without explicitly mentioning.
\begin{fact}
\label{fact:relprop}
\begin{enumerate}
\item Relative entropy is jointly convex in its arguments, that is for distributions $\lambda_1, \lambda_2,\mu_1, \mu_2$
$$ S(p\lambda_1 + (1-p)\lambda_2~ ||~ p\mu_1 + (1-p) \mu_2) \leq p\cdot S(\lambda_1 || \mu_1) + (1-p) \cdot S(\lambda_2 || \mu_2) \enspace .$$
\item Let $XY, X^1Y^1 \in \mcX \times \mcY$. Relative entropy satisfies the following chain rule, 
$$S(X Y || X^1 Y^1 ) = S(X || X^1) + \expct_{x \leftarrow X} S(Y_x || Y^1_x) \enspace .$$
This in-particular implies, using joint convexity of relative entropy,
\begin{align*}
S(X Y || X^1 \otimes Y^1 ) & = S(X || X^1)  +  \expct_{x \leftarrow X} S(Y_x || Y^1) \geq S(X || X^1)  + S(Y || Y^1)  \quad .
\end{align*}
\item For distributions $\lambda, \mu$ : $||\lambda - \mu ||_1 \leq \sqrt{S(\lambda || \mu)}$  and $S(\lambda || \mu) \geq 0$.
\end{enumerate}
\end{fact}
The relative min-entropy between $\lambda$ and $\mu$ is defined as $S_\infty(\lambda || \mu) = \max_{x \in \mcX} \log \frac{\lambda(x)}{\mu(x)}$. It is easily seen that $S(\lambda || \mu) \leq S_\infty(\lambda || \mu)$.
Let $X,Y,Z$ be random variables. The mutual information between $X$ and $Y$ is defined as 
$$I(X : Y) = S(X) + S(Y) - S(XY) = \expct_{x \leftarrow X} S(Y_x || Y) =  \expct_{y \leftarrow Y} S(X_y || X). $$ 
The conditional mutual information is defined as $I(X:Y |~ Z) = \expct_{z \leftarrow Z} I(X : Y |~ Z=z) $. Random variables $XYZ$ form a Markov chain $Z \leftrightarrow X \leftrightarrow Y$ iff $I(Y : Z  |~ X =x) = 0$ for each $x$ in the support of $X$.

\subsubsection*{One-way communication complexity} 
Let $f \subseteq \mcX \times \mcY \times \mcZ$ be a relation. We only consider complete relations that is for each $(x,y) \in \mcX \times \mcY$, there exists at least one $z \in \mcZ$ such that $(x,y,z) \in f$. In the one-way model of communication
there is a single message, from Alice with input $x \in \mcX$ to Bob with input $y \in \mcY$, at the end of which
Bob is supposed to determine an answer $z$ such that $(x,y,z) \in f$. Let~$\eps > 0$ and let $\mu \in \mcX
\times \mcY$ be a distribution. We let $\sD_{\eps}^{1,\mu}(f)$ represent the distributional one-way communication complexity of $f$ under $\mu$
with expected error $\epsilon$, i.e., the communication of the best deterministic one-way protocol for $f$, with distributional
error  (average error over  the inputs) at most
$\eps$ under $\mu$.  Let $\sR^{1,\pub}_{\epsilon}(f)$ represent the public-coin one-way communication complexity of $f$ with worst case error
$\eps$, i.e., the communication of the best public-coin one-way protocol for $f$ with error for each input $(x,y)$ being at
most~$\eps$.  The following is a
consequence of the  min-max theorem in game
theory~\cite[Theorem~3.20, page~36]{KushilevitzN97}.
\begin{lemma}[Yao principle]
\label{lem:yao} $\sR^{1,\pub}_{\epsilon}(f) = \max_{\mu}
\sD_{\epsilon}^{1,\mu}(f)$.
\end{lemma}
The following result follows from the arguments in Braverman and Rao~\cite{BravermanR10}. We skip its proof.
\begin{lemma}[Braverman and Rao~\cite{BravermanR10}]
\label{lem:compress}
Let $f \subseteq \mcX \times \mcY \times \mcZ$ be a relation and $\eps >0$. Let $XY \sim \mu$ be inputs to a  private coins  one-way communication protocol $\mcP$ with distributional error at most $\eps$. Let $M$ represent the message of $\mcP$. Let $\theta$ be the distribution of $XYM$ and let 
$$\Pr_{(x,y,i) \leftarrow \theta} \left[ \log \frac{\theta(i|x)}{\theta(i|y)} > c \right] \leq  \delta .$$ There exists a deterministic one-way protocol $\mcP_1$ for $f$ with inputs distributed according to $\mu$, such that the communication of $\mcP_1$ is $c+O(\log(1/\delta) )$, and distributional error of $\mcP_1$ is at most $\eps + 2\delta$.
\end{lemma}

\section{New bound}
\label{sec:newbounds}
Let $f \subseteq \mcX \times \mcY \times \mcZ$ be a relation, $\mu , \lambda \in \mcX \times \mcY$ be distributions and $\eps, \delta > 0$. 
\begin{definition}[One-way distributions]
Distribution $\lambda$  is called one-way for distribution $\mu$ if for all $(x,y)$ in the support of $\lambda$ we have $\mu(y | x) = \lambda(y | x)$.
\end{definition}

\begin{definition}[Error of a distribution]
Error of  distribution $\mu$ with respect to $f$, denoted $\err_f(\mu)$, is defined as 
 $$\err_f(\mu) \defeq \min \{\Pr_{(x,y) \leftarrow \mu}[(x,y,g(y)) \notin f] ~|~ g : \mcY \rightarrow \mcZ \} \enspace .$$
\end{definition}

\begin{definition}[Robust conditional relative min-entropy]
The $\delta$-robust conditional relative min-entropy of $\lambda$ with respect to $\mu$, denoted $\cment^{\mu}_\delta(\lambda)$, is defined to be the minimum number $c$ such that 
$$ \Pr_{(x,y) \leftarrow \lambda} \left[ \log \frac{\lambda(x|y)}{\mu(x|y )} > c \right]  \leq \delta .$$
\end{definition}

\begin{definition}[Robust conditional relative min-entropy bound]
The  $\eps$-error $\delta$-robust conditional relative min-entropy bound of $f$ with respect to distribution $\mu$, denoted $\cment^\mu_{\eps, \delta}(f)$, is defined as 
$$ \cment^\mu_{\eps, \delta}(f) \defeq \min \{ \cment^\mu_\delta(\lambda) |~ \lambda \mbox{ is one-way for $\mu$ and $\err_f(\lambda) \leq \eps$} \} \enspace. $$ 
The  $\eps$-error $\delta$-robust conditional relative min-entropy bound of $f$, denoted $\cment_{\epsilon,\delta}(f)$, is defined as
$$ \cment_{\eps, \delta}(f) \defeq \max \{ \cment^\mu_{\eps, \delta}(f) |~ \mu \mbox{ is a distribution over } \mcX \times \mcY \} \enspace. $$ 
\end{definition}

The following bound was defined in~\cite{JainKN08} where it was referred to as the one-way subdistribution bound. We call it differently here for consistency of nomenclature with the other bound.
\begin{definition}[Relative min-entropy bound]
The  $\eps$-error relative min-entropy bound of $f$ with respect to distribution $\mu$, denoted $\ment^\mu_\eps(f)$, is defined as 
$$ \ment^\mu_\eps(f) \defeq \min \{ S_\infty(\lambda || \mu) |~ \lambda \mbox{ is one-way for $\mu$ and $\err_f(\lambda) \leq \eps$} \} \enspace. $$ 
The  $\eps$-error relative min-entropy bound of $f$, denoted $\ment(f)$, is defined as
$$ \ment_\eps(f) \defeq \max \{ \ment^\mu_\eps(f) |~ \mu \mbox{ is a distribution over } \mcX \times \mcY \} \enspace. $$ 
\end{definition}
The following is easily seen from definitions.
\begin{lemma}
\label{lem:relations}
$\cment^\mu_\delta(\lambda) \leq S_\infty(\lambda || \mu) $ and hence $\cment^\mu_{\eps,\delta}(f) \leq \ment^\mu_\eps(f)$  and $\cment_{\eps,\delta}(f) \leq \ment_\eps(f)$.
\end{lemma}

\section{New characterization of public coin one-way communication complexity}
\label{sec:newchar}

The following lemma appears in~\cite{JainKN08} . 
\begin{lemma}
\label{lem:Dgeqcrent}
Let $f \subseteq \mcX \times \mcY \times \mcZ$ be a relation and $ \mu \in \mcX \times \mcY$ be a distribution and $\eps, k > 0$. Then,
\[
\sD^{1,\mu}_{\epsilon(1 - 2^{-k})}(f) 
    \quad \geq \quad \ment^\mu_\eps(f)  - k.
\]
\end{lemma}
\suppress{
\begin{proof}
Let $\mcP$ be a deterministic one-way protocol for $f$ with the inputs $XY \sim \mu$. Let $\mcP$ have distributional error at most $\eps$  and communication at most $c = \sD^\mu_\eps(f)$.  Let the random variable $M$ represent the message of $\mcP$ such that $M \leftarrow X \leftarrow Y$. Let $\mu_m$ be the residual distribution on the inputs conditioned on $M=m$. It is clear that $\mu_m$ is one-way for $\mu$. We have $\expct_{m \leftarrow M} [\err_f(\mu_m)] \leq \eps$. Also,
$$ \expct_{m \leftarrow M} [S_\infty(\mu_m || \mu)] = \sum_m \Pr[M=m] \log \frac{1}{\Pr[M=m]} = S(M) \leq c , $$
Therefore, using Markov's inequality, we get an $m$ such that $S_\infty(\mu_m || \mu) \leq 2c$ and $\err_f(\mu_m) \leq 2\eps$.  We conclude desired inequality from the definition of $\ment^\mu_{2\eps}(f)$.
\end{proof}
}
We show the following lemma which we prove later. 
\begin{lemma}
\label{lem:Dleqcrent}
Let $f \subseteq \mcX \times \mcY \times \mcZ$ be a relation and $ \mu \in \mcX \times \mcY$ be a distribution and $\eps, \delta > 0$. Then,
$$ \sD^{1,\mu}_{\eps + 4\delta}(f) \leq  \cment_{\eps,\delta}(f)  + O(\log \frac{1}{\delta}) \enspace .$$
\end{lemma}

\begin{theorem}
\label{thm:newchar}
Let $f \subseteq \mcX \times \mcY \times \mcZ$ be a relation and $\eps > 0$. Then,
\begin{align*}
\ment_{2\eps}(f) -1 \leq \sR_\eps^{1,\pub} (f) & \leq \cment_{\eps/5, \eps/5}(f) + O(\log \frac{1}{\eps}) \enspace .
\end{align*}
Hence $$ \sR_\eps^{1,\pub}(f) = \Theta(\ment_{\eps}(f) ) = \Theta(\cment_{\eps,\eps}(f) ) \enspace .$$
\end{theorem}

\begin{proof}
The first inequality follows from Lemma~\ref{lem:Dgeqcrent} (set $k=1$) and maximizing both sides over all distributions $\mu$ and using Lemma~\ref{lem:yao}. The second inequality follows from Lemma~\ref{lem:Dleqcrent} (set $\eps = \eps, \delta= \eps$) and maximizing both sides over all distributions $\mu$ and using Lemma~\ref{lem:yao}. The other relations now follow from Lemma~\ref{lem:relations} and from the fact that the error in public coin randomized one-way communication complexity can be made a constant factor down by increasing the communication by a constant factor.
\end{proof}

\begin{proofof}{Lemma~\ref{lem:Dleqcrent}}
We make the following key claim which we prove later.
\begin{claim}
\label{claim:key}
There exists a natural number $k$ and a Markov chain $M \leftrightarrow X \leftrightarrow Y$, where $M \in [k]$ and $XY \sim \mu$, such that 
\begin{enumerate}
\item  for each $i\in [k] ~ : ~ \err_f(P_i) \leq \eps$, where $P_i = (XY |~ M=i)$,
\item  $\Pr_{(x,y,i) \leftarrow \theta} \left[ \log \frac{\theta(i|x)}{\theta(i|y)} > \cment_{\eps,\delta}(f) + \log \frac{1}{\delta} \right] \leq  2\delta $, where $\theta$ is the distribution of $XYM$.   
\end{enumerate}
\end{claim}
The above claim immediately gives us a private-coin one-way prootocol $\mcP_1$ for $f$, where Alice on input $x$ generates $i$ from the distribution $M_x$ and sends $i$ to Bob. It is easily seen that the distributional error of $\mcP_1$ is at most $\eps$. Now using Lemma~\ref{lem:compress} we get a deterministic protocol $\mcP_2$ for $f$, with distributional error at most $\eps + 4\delta$ and communication at most $d =  \cment_{\eps,\delta}(f)  + O(\log \frac{1}{\delta})$. 
\end{proofof}

We return to proof of Claim~\ref{claim:key}.

\begin{proofof}{Claim~\ref{claim:key}}
Let $c = \cment_{\eps,\delta}(f)$. Let us perform a procedure as follows. Start with $i = 1$.
\begin{enumerate}
\item \label{step:repeat} Let us say we have collected distributions $P_1, \ldots, P_{i-1}$, each one-way for $\mu$, and positive numbers $p_1, \ldots, p_{i-1}$ such that  $\mu \geq \sum_{j=1}^{i-1} p_j P_j$. If  $\mu = \sum_{j=1}^{i-1} p_j P_j$  then set $k = i-1$ and stop. 
\item Otherwise let us express $ \mu = \sum_{j=1}^{i-1} p_j P_j + q_iQ_{i}$, where $Q_{i}$ is a distribution,  one-way for $\mu$.
Since $\cment_{\eps,\delta}^{Q_i}(f) \leq c$, we know that there is a distribution $R$, one-way for $Q_{i}$ (hence also one-way for $\mu$), such that $\cment^{Q_{i}}_\delta(R) \leq c$ and $\err_f(R) \leq \eps$. Let $r = \max\{ q |~ Q_{i} \geq q R \}$. Let $P_{i} = R, p_{i} = q_{i} * r, i=i+1$ and go back to step~\ref{step:repeat}.
\end{enumerate}

It can be observed that for each new $i$, there is a new $x \in \mcX$ such that $Q_i(x) = 0$. Hence the above process converges after at most $|\mcX|$ iterations.
At the end we have $\mu = \sum_{i=1}^{k} p_i P_i$.

Let us define $M\in [k]$ such that $\Pr[M=i] = p_i$. Let us define $XY \in \mcX \times \mcY$ correlated with $M$ such that $(XY | ~ M=i) \sim P_i$. It is easily checked that $XY \sim \mu$. Also since each $P_i$ is one-way for $\mu$, $XYM$ form a Markov chain $M \leftrightarrow X \leftrightarrow Y$.   Let $\theta$ be the distribution of $XYM$.
Let us define 
\begin{enumerate}
\item  $B = \{ (x,y,i) |~ \log \frac{P_i(x|y)}{\mu(x|y)} > c + \log \frac{1}{\delta} \}$, 
\item  $B_1 = \{ (x,y,i) |~ \log \frac{P_i(x|y)}{Q_i(x|y)} > c\},$
\item  $B_2 = \{ (x,y,i) |~ \frac{\mu(y)}{q_i Q_i(y)} > \frac{1}{ \delta} \} .$
\end{enumerate}
Since $q_iQ(x,y) \leq \mu(x,y)$,
\begin{align*}
\frac{P_i(x|y)}{\mu(x|y)} & = \frac{P_i(x|y)}{Q_i(x|y)} \cdot \frac{Q_i(x|y)}{\mu(x|y)} = \frac{P_i(x|y)}{Q_i(x|y)} \cdot \frac{Q(x,y) \mu(y)}{Q(y) \mu(x,y)} \leq \frac{P_i(x|y)}{Q_i(x|y)} \cdot \frac{\mu(y)}{q_i Q(y)} 
\end{align*}
Therefore $B \subseteq B_1 \cup B_2$. Since for each $i, ~ \cment^{Q_{i}}_\delta(P_i) \leq c$, we have 
$$\Pr_{(x,y,i) \leftarrow \theta}[(x,y,i) \in B_1] \leq \delta .$$
For a given $y$, let $i_y$ be the smallest $i$ such that $\frac{\mu(y)}{q_{i} Q_i(y)} > \frac{1}{ \delta}$. Then,
$$\Pr_{(x,y,i) \leftarrow \theta}[(x,y,i) \in B_2] = \sum_y q_{i_y} Q_{i_y}(y) < \sum_y \delta \mu(y) = \delta .$$ 
Hence, $\Pr_{(x,y,i) \leftarrow \theta}[(x,y,i) \in B] < 2\delta $. Finally note that,
\begin{align*}
\frac{P_i(x|y)}{\mu(x|y)}  & =  \frac{\theta(x|(y,i))}{\theta(x|y)} = \frac{\theta(x|y)\theta(i|(x,y))}{\theta(i|y)\theta(x|y)} = \frac{\theta(i|x)}{\theta(i|y)}  \enspace .
\end{align*}
\end{proofof}

\section{Strong direct product for one-way communication complexity}
\label{sec:dpt}

We start with the following theorem which we prove later. 
\begin{theorem}[Direct product in terms of $\ment$ and $\cment$]
\label{thm:dptment}
Let $f \subseteq \mcX \times \mcY \times \mcZ$ be a relation and $\mu \in \mcX \times \mcY$ be a distribution. Let $0 < 200 \sqrt{\delta} < \eps < 0.5$ and $k$ be a natural number.  Then
$$ \ment^{\mu^k}_{1 - (1- \eps/2)^{\floor {\delta k}}}(f^k)  \geq \delta \cdot k \cdot \cment^\mu_{\eps,\eps}(f) \enspace .$$
\end{theorem}
We now state and prove our main result.
\begin{theorem}[Direct product for one-way communication complexity]
\label{thm:dpt}
Let $f \subseteq \mcX \times \mcY \times \mcZ$ be a relation. Let $0 < 200\sqrt{\delta} < \eps < 0.5$  and $k$ be a natural number.   Let $\delta' =  (1- \eps/10)^{\floor {\delta k}} + 2^{-k}$. There exists a constant $\kappa$ such that, 
$$ \sR^{1,\pub}_{1 - \delta'}(f^k)  \geq  \frac{\delta \cdot k}{\kappa} \cdot \sR^{1,\pub}_\eps(f)  - k\enspace .$$
In other words,
$$ \sR^{1,\pub}_{1 - 2^{-\Omega(k)}}(f^k)  \geq   \Omega( k \cdot \sR^{1,\pub}_{1/3}(f)  )  \enspace . $$
\end{theorem}
\begin{proof}
Let $\mu_1$ be a distribution such that $\sD^{1,\mu_1}_\eps(f) = \sR^{1,\pub}_\eps(f)$. Let $\mu$ be a distribution such that $\cment_{\eps/5,\eps/5}^{\mu}(f) = \cment_{\eps/5,\eps/5}(f)$.  Let $\kappa$ be a constant (guaranteed by Lemma~\ref{lem:Dleqcrent}) such that $ \sD^{1,\mu_1}_\eps(f)  \leq \kappa \cdot \cment_{\eps/5,\eps/5} (f) $.
Using Lemma~\ref{lem:Dgeqcrent}, Lemma~\ref{lem:Dleqcrent} and Theorem~\ref{thm:dptment},
\begin{align*}
\frac{\delta \cdot k}{\kappa} \cdot \sR^{1,\pub}_\eps(f) &= \frac{\delta \cdot k}{\kappa} \cdot \sD^{1,\mu_1}_\eps(f) \\
& \leq \delta \cdot k \cdot \cment_{\eps/5,\eps/5} (f)  = \delta \cdot k  \cdot \cment^{\mu}_{\eps/5,\eps/5} (f) \\
& \leq  \ment^{{\mu}^k}_{1 - (1- \eps/10)^{\floor {\delta k }}}(f^k)  \leq \sD^{1,{\mu}^k}_{1 - (1- \eps/10)^{\floor {\delta k  }} - 2^{-k}}(f^k) + k \\
& \leq \sR^{1,\pub}_{1 - \delta'}(f^k) + k \enspace .
\end{align*}
\end{proof}

\begin{proofof}{Theorem~\ref{thm:dptment}}
Let $c = \cment^\mu_{\eps,\eps}(f) $. Let $\lambda \in \mcX^k \times \mcY^k$ be a distribution which is one-way for $\mu^k$ and with $S_\infty (\lambda || \mu^k) < \delta c k$. We show that $\err_{f^k}(\lambda) \geq 1 - (1 - \eps/2)^{\floor{ \delta k  }}$. This shows the desired. 

Let $B$ be a set. For a random variable distributed in $B^k$, or a string in $B^k$, the portion corresponding to the $i$th coordinate is represented with subscript $i$. Also the portion except the $i$th coordinate is represented with subscript $-i$. Similarly portion corresponding to a subset $C \subseteq [k]$ is represented with subscript $C$. For joint random variables $MN$, we let $M_n$ to represent $M |~ (N=n)$ and also $M N |~ (N=n)$ and is clear from the context.

Let $X Y \sim \lambda$. Let us fix $g: \mcY^k \rightarrow \mcZ^k$. For a coordinate $i$, let the binary random variable $T_i \in \{0,1\}$, correlated with $XY$, denote success in the $i$th coordinate.  That is $T_i = 1$ iff $XY = (x,y)$ such that $(x_i, y_i, g(y)_i) \in f$. 
We make the following claim which we prove later. Let $k' = \floor {\delta k  }$.
\begin{claim}
\label{claim:succ}
There exists $k'$ distinct coordinates $i_1, \ldots ,i_{k'}$  such that $\Pr[T_{i_1} = 1] \leq 1 - \eps/2$ and for each  $r < k' $,
\begin{enumerate}
\item either $\Pr[T_{i_1} \times  T_{i_2} \times \cdots \times T_{i_r} = 1 ] \leq (1 - \eps/2)^{k'}$,
\item or  $\Pr[T_{i_{r+1}} = 1 |~ (T_{i_1} \times  T_{i_2} \times \cdots \times T_{i_r} = 1)] \leq 1 - \eps/2 $. 
\end{enumerate}
\end{claim}
This shows that the overall success is 
$$ \Pr[T_1 \times T_2 \times \cdots \times T_k  = 1] \leq \Pr[T_{i_1} \times T_{i_2} \times \cdots \times T_{i_{k'}}  = 1] \leq (1-\eps/2)^{k'} \enspace .$$
\end{proofof}

\begin{proofof}{Claim~\ref{claim:succ}} 
\suppress{
We start by identifying $i_1$. Consider,
\begin{align*}
\delta c k  > S_\infty(\lambda || \mu^k) 
& \geq \crent^{\mu^k}(\lambda) = \expct_{y \leftarrow Y}  S(X_{y}  || X'_{y_1} \otimes \ldots \otimes X'_{y_k}  ) \\
& \geq \expct_{y \leftarrow Y}  \sum_{i=1}^k S((X_y)_i  || X'_{y_i}  ) =\sum_{i=1}^k  \expct_{y \leftarrow Y}   S((X_y)_i  || X'_{y_i}  )  \enspace .
\end{align*}
Therefore there exists a coordinate $j$ such that $ \expct_{y \leftarrow Y}   S((X_y)_{j}  || X'_{y_{j}}) \leq \delta c$. Let 
$$G = \{ y_{-j}  | ~  \left( \expct_{y_{j} \leftarrow Y_{j} | (Y_{-j} = y_{-j})} S((X_y)_{j}  || X'_{y_{j}}) \right )  \leq c\} \enspace .$$ 
Therefore $\Pr \left[  Y_{-j}  \in G \right] \geq 1 -\delta \enspace .$ Fix $ y_{-j} \in G$. It can be observed that $(XY)_j | ~(Y_{-j} = y_{-j}) $ is one-way for $\mu$. Therefore $\Pr[T_j = 1 |~  (Y_{-j} = y_{-j}) ]  \leq 1 - \eps$.
Hence overall,
$$\Pr[T_j = 1 ]  \leq (1 - \eps)(1-\delta) + \delta  \leq 1 - \eps/2 \enspace . $$ 
Set $i_1 = j$. }
Let us say we have identified $r < k'$ coordinates $i_1, \ldots i_{r}$.  Let $C = \{ i_1, i_2, \ldots, i_r\}$. Let $T = T_{i_1} \times T_{i_2} \times \cdots \times T_{i_r}$ . If $\Pr[T=1] \leq (1- \eps/2)^{k'}$ then we will be done. So assume that $ \Pr[T=1]  > (1 - \eps/2)^{k'} \geq 2^{-\delta k } $. 

Let $X' Y' \sim \mu$. Let $X^1 Y^1  = (XY | ~ T = 1)$. Let  $D$ be uniformly distributed  in $\{0,1\}^k$ and independent of $X^1Y^1$. Let $U_i = X^1_i$ if $D_i=0$ and $U_i = Y^1_i$ if $D_i=1$.  Let $U = U_1 \ldots U_k$. Below for any random variable $\tilde{X}\tilde{Y}$, we let $\tilde{X}\tilde{Y}_{d,u}$, represent the random variable obtained by appropriate conditioning on $\tilde{X}\tilde{Y}$: for all $i$, $\tilde{X}_i = u_i$ if $d_i=0$ otherwise $\tilde{Y}_i = u_i$ if $d=1$ . 
 Consider,
\begin{align}
 \delta k + \delta c k   &>  S_\infty(X^1Y^1  || XY)  + S_\infty(XY || (X'Y')^{\otimes k} ) \nonumber \\
& \geq  S_\infty(X^1Y^1 || (X'Y') ^{\otimes k} )  \geq S(X^1Y^1 || (X'Y') ^{\otimes k} ) = \expct_{d \leftarrow D} S(X^1Y^1 || (X'Y') ^{\otimes k} ) \nonumber\\
& \geq \expct_{(d,u,x_C,y_C) \leftarrow (D U X^1_C Y^1_C)} S((X^1Y^1)_{d,u,x_C,y_C}  || ((X'Y')^{\otimes k})_{d,u,x_C,y_C} ) \nonumber \\
& \geq  \expct_{(d,u,x_C,y_C) \leftarrow (D U X^1_C Y^1_C)} S(X^1_{d,u,x_C,y_C}  || X'_{d_1,u_1,x_C,y_C} \otimes \ldots \otimes X'_{d_k,u_k,x_C,y_C}  )  \nonumber \\
& \geq  \expct_{(d,u,x_C,y_C) \leftarrow (D U X^1_C Y^1_C)} \sum_{i \notin C} S((X^1_{d,u,x_C,y_C})_i  || X'_{d_i,u_i}  ) \nonumber \\
 &= \sum_{i \notin C} \expct_{(d,u,x_C,y_C) \leftarrow (D U X^1_C Y^1_C)}  S((X^1_{d,u,x_C,y_C})_i  || X'_{d_i,u_i}  )   \enspace . \label{eq:1}
\end{align}
Also
\begin{align}
 \delta k &> S_\infty(X^1Y^1 || XY)   \geq S(X^1Y^1 || XY) = \expct_{d \leftarrow D}  S(X^1Y^1 || XY) \nonumber  \\
& \geq \expct_{(d,u,x_C,y_C) \leftarrow (D U X^1_C Y^1_C)} S(Y^1_{d,u,x_C,y_C} ~ || ~Y_{d_1,u_1,x_C,y_C} \otimes \ldots \otimes Y_{d_k,u_k,x_C,y_C}  )  \nonumber  \\
&  \geq \expct_{(d,u,x_C,y_C) \leftarrow (D U X^1_C Y^1_C)} \sum_{i \notin C} S((Y^1_{d,u,x_C,y_C})_i  || Y_{d_i,u_i} )  \nonumber  \\
& =  \sum_{i \notin C} \expct_{(d,u,x_C,y_C) \leftarrow (D U X^1_C Y^1_C)} S((Y^1_{d,u,x_C,y_C})_i  || Y'_{d_i,u_i} )  \enspace .\label{eq:2} 
\end{align}
From Eq.~\ref{eq:1} and Eq.~\ref{eq:2} and using Markov's inequality we get a coordinate $j$ outside of $C$ such that 
\begin{enumerate}
\item $ \expct_{(d,u,x_C,y_C) \leftarrow (D U  X^1_C Y^1_C)}  S((X^1_{d,u,x_C,y_C})_j  || X'_{d_j,u_j}  )  \leq \frac{2 \delta(c+1)}{(1-\delta)} \leq 4 \delta c  ,$ and
\item $ \expct_{(d,u,x_C,y_C) \leftarrow (D U  X^1_C Y^1_C)} S((Y^1_{d,u,x_C,y_C})_j  || Y'_{d_j,u_j} ) \leq \frac{2 \delta }{(1-\delta) } \leq 4 \delta$.
\end{enumerate}
Therefore,
\begin{align*}
4 \delta c  &\geq \expct_{(d,u,x_C,y_C) \leftarrow (D U  X^1_C Y^1_C)}  S((X^1_{d,u,x_C,y_C})_j  || X'_{d_j,u_j}  ) \\
& = \expct_{(d_{-j},u_{-j},x_C,y_C) \leftarrow (D_{-j} U_{-j}  X^1_C Y^1_C)} \expct_{(d_j,u_j) \leftarrow (D_j U_j)  | ~(D_{-j} U_{-j}  X^1_C Y^1_C) = (d_{-j},u_{-j},x_C,y_C)  }   S((X^1_{d,u,x_C,y_C})_j  || X'_{d_j,u_j}  )    .
\end{align*}
 And,
\begin{align*}
4 \delta &\geq \expct_{(d,u,x_C,y_C) \leftarrow (D U X^1_C Y^1_C)}   S((Y^1_{d,u,x_C,y_C})_j  || Y'_{d_j,u_j} ) \\
& = \expct_{(d_{-j},u_{-j},x_C,y_C) \leftarrow (D_{-j} U_{-j} X^1_C Y^1_C)} \expct_{(d_j,u_j) \leftarrow (D_j U_j )  | ~(D_{-j} U_{-j} X^1_C Y^1_C) = (d_{-j},u_{-j},x_C,y_C)  }  S((Y^1_{d,u,x_C,y_C})_j  || Y'_{d_j,u_j} )  .
\end{align*}
Now using Markov's inequality, there exists set $G_1 $ with $\Pr[Y^1_{-j} \in G_1] \geq 1 - 0.2 $,  such that for all $(d_{-j},u_{-j},x_C,y_C) \in G_1$,
\begin{enumerate}
\item \label{it:1} $ \expct_{(d_j,u_j) \leftarrow (D_j U_j )  | ~(D_{-j} U_{-j} X^1_C Y^1_C) = (d_{-j},u_{-j},x_C,y_C)  }     S((X^1_{d,u,x_C,y_C})_j  || X'_{d_j,u_j}  )    \leq 40 \delta c$, \quad and
\item \label{it:2} $ \expct_{(d_j,u_j) \leftarrow (D_j U_j )  | ~(D_{-j} U_{-j} X^1_C Y^1_C) = (d_{-j},u_{-j},x_C,y_C)  }     S((Y^1_{d,u,x_C,y_C})_j  || Y'_{d_j,u_j} ) \leq 40 \delta $.
\end{enumerate}
Fix $(d_{-j},u_{-j},x_C,y_C) \in G_1$. Conditioning on $D_j=1$ (which happens with probability $1/2$) in inequality~\ref{it:1}. above we get,
\begin{equation}
\expct_{y_j \leftarrow Y^1_j  | (D_{-j} U_{-j} X^1_C Y^1_C) = (d_{-j},u_{-j},x_C,y_C)   }   S((X^1_{d_{-j},u_{-j},y_j,x_C,y_C})_j  || X'_{y_j}  )    \leq 80 \delta c .
\label{eq:3}
\end{equation}
Conditioning on $D_j=0$ (which happens with probability $1/2$) in inequality~\ref{it:2}. above we get,
\begin{equation*}
\expct_{x_j \leftarrow X^1_j  | (D_{-j} U_{-j} X^1_C Y^1_C) = (d_{-j},u_{-j},x_C,y_C)   }  S((Y^1_{d_{-j},u_{-j},x_j,x_C,y_C})_j  || Y'_{x_j} )  \leq 80 \delta  .
\end{equation*}
Using concavity of square root we get,
\begin{equation} \expct_{x_j \leftarrow X^1_j  |(D_{-j} U_{-j} X^1_C Y^1_C) = (d_{-j},u_{-j},x_C,y_C)  }  ||(Y^1_{d_{-j},u_{-j},x_j,x_C,y_C})_j  -  Y'_{x_j} ||_1  \leq \sqrt{80 \delta}  .
\label{eq:4}
\end{equation}
Let $X^2Y^2$ be such that $X^2 \sim (X^1_{d_{-j},u_{-j},x_C,y_C})_j$ and $(Y^2 | ~X^2=x_j) \sim  Y'_{x_j}$. From Eq.~\ref{eq:4} we get,
\begin{equation} \label{eq:5} ||X^2Y^2 - ((X^1Y^1)_{d_{-j},u_{-j},x_C,y_C})_j  ||_1 \leq \sqrt{80 \delta} . \end{equation} 
From construction $X^2Y^2$ is one-way for $\mu$.  Using using Eq.~\ref{eq:3} and Eq.~\ref{eq:5}  we conclude that 
$$ \Pr_{(x,y) \leftarrow X^2Y^2}\left[\log \frac{X^2Y^2(x|y)}{\mu(x|y)}  > c\right] \leq 100 \delta + \sqrt{80 \delta} \leq \eps .$$
Hence $\cment^{\mu}_{\eps}(X^2Y^2) \leq c$. Hence, $\err_f(X^2Y^2) \geq \eps $ and therefore
\begin{align*}
\err_f(((X^1Y^1)_{d_{-j},u_{-j},x_C,y_C})_j) \geq   \eps -  \sqrt{80 \delta} \geq \frac{3\eps}{4} .
\end{align*}
Since conditioned on $(Y^1_{d_{-j},u_{-j},x_C,y_C})_j$, the distribution $(X^1Y^1)_{d_{-j},u_{-j},x_C,y_C}$ is product across the $\mcX^k$ and $\mcY^k$ parts, we have, 
\begin{align*}
\Pr[T_j = 1|~ (1,d_{-j},u_{-j},x_C,y_C) = (T  D_{-j} U_{-j} X_C Y_C)] & \leq 1 - \err_f(((X^1Y^1)_{d_{-j},u_{-j},x_C,y_C})_j) .
\end{align*}
Therefore overall 
$$\Pr[T_j  = 1 | ~ (T=1)] \leq 0.8(1 -  \frac{3\eps}{4})  + 0.2  \leq 1 - \eps/2 .$$
\end{proofof}

\bibliographystyle{alpha}

\bibliography{sdpone}

\end{document}